\documentclass[11pt]{article}
\usepackage{amsfonts}
\usepackage{amsmath}
\usepackage{amsthm}
\usepackage{mathabx}
\usepackage{sgame}
\usepackage{graphicx}
\usepackage{pgfplots}
\pgfplotsset{compat=1.7}
\usetikzlibrary{intersections}
\usetikzlibrary{calc}
\usetikzlibrary{patterns}
\usepackage[normalem]{ulem}
\usepgfplotslibrary{fillbetween}
\usetikzlibrary{patterns}
\usepackage{comment}

\newtheorem{definition}{Definition}
\newtheorem{proposition}{Proposition}

\newtheorem{lemma}{Lemma}

\usepackage[textwidth=30mm]{todonotes}
\usepackage{enumerate}

\usepackage[colorlinks=true,breaklinks=true,bookmarks=true,urlcolor=blue,
     citecolor=blue,linkcolor=blue,bookmarksopen=false,draft=false]{hyperref}

\newcommand{\YT}[1]{\todo[color=red!50,author=\textbf{Yevgeny},inline]{\small #1\\}}
\newcommand{\GF}[1]{\todo[color=yellow!20,author=\textbf{Ga\"etan},inline]{\small #1\\}}
\newcommand{\AG}[1]{\todo[color=white, author=\textbf{Alberto},inline]{\small #1\\}}

\def\E{\mathbb{E}} 
\def\1{{1\hskip-2.5pt{\rm l}}} 
\def\a{\alpha} 
\DeclareMathOperator*{\argmax}{arg\,max}

\usepackage[square]{natbib}

\begin{document}

\title{{Strategic flip-flopping in political competition}\thanks{The authors wish to thank the participants of the Game Theory seminar in IHP, the eco-lunch seminar in AMSE, the Workshop on Dynamic Games and Applications in Paris 2 and the Workshop on Stochastic Methods in Erice for their highly valuable comments.
YT acknowledges the support of the Israel Science Foundation grants \#2566/20, \#1626/18 \& \#448/22. GF thanks the “France 2030” investment plan managed by the French National Research Agency (reference :ANR-17-EURE-0020) and from Excellence Initiative of Aix-Marseille University -- A*MIDEX.
}}
\author{
Ga\"etan Fournier\thanks{Aix Marseille Univ, CNRS, AMSE, Marseille, France. \textsf{gaetan.fournier@univ-amu.fr}},  Alberto Grillo\thanks{Université Paris-Panthéon-Assas, LEMMA, Paris, France. \textsf{alberto.grillo@u-paris2.fr}},  Yevgeny Tsodikovich\thanks{Department of Economics, Bar Ilan University, Israel. \textsf{yevgets@gmail.com}}}

\maketitle

\thispagestyle{empty}

\begin{abstract}
We study candidates' positioning when adjustments are possible in response to new information about voters' preferences. Re-positioning allows candidates to get closer to the median voter but is costly both financially and electorally. 
We examine the occurrence and the direction of the adjustments depending on the ex-ante positions and the new information. In the unique subgame perfect equilibrium, candidates anticipate the possibility to adjust in response to future information and diverge ex-ante in order to secure a cost-less victory when the new information is favorable.
\end{abstract}



\newpage
\lineskip=1.8pt\baselineskip=18pt\lineskiplimit=0pt \count0=1

\section{Introduction}
In the run-up to an election, candidates largely rely on polls to learn about voting intentions. Uncovering the electorate's leanings, what is known as the \textit{policy mood} \citep{stimson1991public, stevenson2001economy}, help them shape their campaign message but may also pose a difficult choice: should they change their positions to get closer to voters? In the long run, it is well documented that politicians adapt their stances to reflect the evolving preferences of their constituencies \citep{glazer1985congressional,stratmann2000congressional,miler2016legislative}. Yet, policy changes that are too sudden involve substantial costs, in terms of not only communication efforts but also electoral appeal. In the history of the U.S. presidential elections, for example, the defeats of John Kerry in 2004 and Matt Romney in 2012 are often linked to their shifting views on salient issues\footnote{Specifically concerning Kerry's opposition to the war in Iraq after his previous support and Romney's multiple shifts, notably on abortion \citep{croco2016flipside}.}.

This paper studies re-positioning choices as a strategic game between candidates. We assume that voters prefer candidates with platforms aligned with their ideal policies, but dislike \textit{flip-floppers}, i.e. candidates who strategically change their positions during the campaign. We examine which candidates are more likely to adjust their positions following new information, in which direction, and how successfully. We also investigate how the anticipation of possible changes affects candidates' positions ex-ante, before the information is revealed. 

A theory of policy adjustments is useful in light of the bias of the empirical evidence. As pointed out by \cite{tomz2012political}, \textit{``historical data [...] reveal the consequences of re-positioning only in the specific circumstances when politicians thought re-positioning would be optimal"}. Our analysis of the involved trade-offs aims to clarify what these specific circumstances are.

\textbf{Model.}
We enrich the Downs-Hotelling framework by supposing an information shock creating a two-stage game. The shock reveals the location of the median voter. This captures the idea that voters' aggregate preferences fluctuate over time and that their current leanings are disclosed during the electoral campaign. In the model, two office-motivated candidates first select their positions before the shock, knowing only the distribution of the median voter's location. They can then revise their positions after the shock (after learning the actual location). 
We assume that such a change involves both an \emph{electoral cost} -- voters' discount their evaluation of a candidate who changed position -- and an \emph{organizational cost} -- the candidate's payoff is reduced due to the policy change.

On the one hand, the electoral cost represents voters' negative feeling toward a flip-flop. This aversion may be interpreted as a higher uncertainty about what the candidate would do if elected, i.e. to an undermined credibility of the position \citep{enelow1993elements}. Alternatively, voters may value consistency on policy issues as a cue for character \citep{kartik2007signaling}, or as a signal for quality of implementation of the ex-post policy. We model voters' dislike for flip-flops in a reduced form, through a penalty in their utility if a candidate changes position.

On the other hand, the organizational cost represents all other costs involved by a change of position which are not related to voters' reaction. The most prominent is the organizational and financial cost for candidates of communicating the change to the voters. While justifying a policy change to the public may reduce its electoral harm, such communication requires costly advertising. 

\textbf{Results.} 
We study the subgame-perfect equilibria of the game, using backward induction. In the second stage, if the revealed information is not clearly in favor of one candidate, the adjustment choices determine the winner of the election. The advantaged candidate would like to adjust only if the other candidate threatens the advantage by also moving, while the disadvantaged candidate would like to adjust only if his opponent does not. The only equilibrium is therefore in mixed strategies, with both candidates flip-flopping with positive probability and having a chance to win the election.  If instead the revealed information strongly favors one candidate, there are no incentives to flip-flop as the election cannot be disputed anymore.

In the first stage, candidates anticipate their strategic responses to the information shock. The game has at most one subgame-perfect equilibrium, which exists if both the electoral and the organizational costs are sufficiently high. In such equilibrium, candidates contradict the median voter theorem: they choose differentiated platforms to secure a cost-less victory when the information is in their favor. 

\textbf{Contribution.} Assuming sufficiently high adjustment costs, our model yields several implications concerning the occurrence of flip-flops. First, along the equilibrium path, re-positioning happens only toward the center, while candidates never adjust toward more extreme positions. This highlights a \textit{moderation effect}, according to which candidates cultivate separate electorates when elections are far in time and then soften their positions during the campaign. In the U.S., such an effect is often attributed to the presence of primary elections, in which candidates need to convince a more extreme median voter \citep{agranov2016flip}. Our framework provides a different rationale for a similar dynamic, which can play out even in the absence of primary elections.

A second prediction is that flip-flops consist mostly of small adjustments made by an advantaged candidate in order to secure his victory.  Only a minority of flip-flops are large adjustments made by a disadvantaged candidate who seeks to reverse the likely outcome of the election. Indeed, on the one hand, when the favorite candidate adjusts his position, the magnitude of the adjustment is smaller than when his opponent adjusts. On the other hand, we find that a candidate favored by the new information is more likely to adjust his position than a disadvantaged candidate. Hence, an adjustment by the advantaged candidate guarantees his victory but an adjustment by the challenger is more likely to be unsuccessful. 

Finally, we provide comparative statics results with respect to changes in the adjustment costs. We find that an increase in the electoral cost decreases the polarization of candidates but increases their equilibrium payoffs. Indeed, such an increase makes flip-flopping less likely, because it increases the probability that the election is secured after the information shock, which guarantees the favorite candidate a cost-less victory.

\textbf{Extension.} We look at how the results are modified by an asymmetry between candidates. Asymmetry in organizational costs has no impact on the results. Asymmetry in electoral costs makes the more flexible candidate choose a more central platform, while the less flexible candidate offers a more polarized platform. We show that a candidate who is significantly less flexible in adjusting his position loses the election in equilibrium, even if he is favored by the information on voters. In the general asymmetric game, payoffs are decreasing in candidates' own electoral cost, but increasing in the opponent's electoral cost.

\textbf{Literature.} Following the seminal work by \cite{hotelling1929stability} and \cite{downs1957economic}, most models of political competition assume that candidates can freely commit to any electoral platform. At the extreme opposite, some papers take campaign announcements as cheap-talk and assume that, once elected, candidates act according to their own preferences \citep{alesina1988credibility,osborne1996model,besley1997economic}.

We lay down a more realistic framework, in which platforms are binding although not immutable, and candidates can costly adjust them over time. The evidence that politicians respond in an adaptive way to changes in voters' preferences is substantial in political science, see \cite{stimson1995dynamic}, \cite{adams2004understanding}, and \cite{kousser2007ideological}. \cite{karol2009party} argues that elite replacement is not necessary for parties' policy changes, which are often driven by incumbent politicians. \cite{adams2006niche} find that mainstream parties engage in re-positioning more than niche parties, while \cite{tavits2007principle} relates the likelihood of success to whether the change concerned pragmatic or principled issues. This literature has a strong empirical focus and recently turned also to experimental settings \citep{tomz2012political,doherty2016changing,robison2017role}. We see our theoretical investigation as a useful complement to this existing research.

Key to our analysis is the assumption that policy changes are costly. Models by \cite{bernhardt1985candidate}, \cite{ingberman1989reputational} and \cite{enelow1993elements} consider a voters' utility function which is decreasing in the size of a change in candidates' policies. In these studies, the loss is derived from the increased uncertainty concerning the policy that will be effectively implemented. Evidence of an intrinsic preference of voters for consistency is also discussed as a ``waffle effect'' in political psychology \citep{carlson1985waffle}. \cite{hoffman1984political} found that even proposing a policy in agreement with the voters' preferences may not be rewarded if it follows an inconsistent track-record. Our specification is consistent with \cite{debacker2015flip}, who finds empirically that the electoral cost is increasing in the size of the change. The organizational cost, instead, captures in a reduced form the idea that informative advertising directed towards voters is costly for candidates \citep{coate2004political,ashworth2006campaign}. Our model abstracts away the role of interest groups in financing political campaigns, and assumes for simplicity that candidates bear communications costs themselves.\footnote{See \cite{prato2019campaign} for a similar assumption.}

Given the result of ex-ante divergence, our paper belongs to a class of models in which policy differentiation is adopted to soften competition in a second dimension. The baseline argument is familiar in industrial organization models, in which firms differentiate their products in order to reduce the subsequent competition in prices \citep{tirole1988theory,shaked1982relaxing}. In political competition, \cite{ashworth2009elections} and \cite{zakharov2009model} examine the reduced need for differentiated candidates to invest in valence. \cite{eyster2007party} consider parties as collections of candidates. Parties choose general platforms first and candidates can deviate from their party’s platform at a cost, in order to increase their chances in their specific constituency. Divergence of parties’ platforms results from the incentive to minimize the aggregate cost of candidates’ re-positioning for the party. In \cite{balart2022technological}, candidates diverge to minimize future costly advertising, which would be needed to impress voters if the ideological divide were low. We focus instead on candidates’ choice of changing their own position over time, when doing this is not only financially costly but also penalized electorally. While ex-ante divergence arises analogously to these previous models in the literature, our framework allows us to highlight the emerging patterns of flip-flopping along the (subgame-perfect) equilibrium.

The paper is more loosely related to three other research lines. First, flip-flopping has been investigated in two-stage models with primary elections \citep{hummel2010flip,agranov2016flip}. Although we do not have primary elections, our model exhibits a similar dynamic of moderation during the campaign stage. Second, the paper shares with \cite{kamada2020optimal} an interest in the dynamic aspect of policy announcements during the campaign stage. In their paper, however, the opportunity to revise a position arises stochastically in a revision game \`a la \cite{kamada2020revision}, and the focus is on candidates' reactions to each others' announcements. Finally, policy adjustments in our model can be seen as a pandering phenomenon. However, such pandering only describes a spatial movement and ignores considerations of optimality from a common-value perspective, addressed in \cite{maskin2004politician} and \cite{andreottola2021flip}.

\textbf{The rest of the paper is organized as follows.}
Section~\ref{Sect:Model} presents the model and Section~\ref{Sect:MainResult} the equilibrium analysis and the results. We study the case of asymmetric candidates in Section~\ref{Sect:ASYM_candidates}. Technical proofs and lemmas are postponed to the appendix.

\section{The Model}\label{Sect:Model}

There are two office seeking candidates (namely, candidate~$1$ and $2$) and a continuum of voters that we identify with their ideal policy $t \in \mathbb{R}$.  The location of the median voter $m$ is drawn according to a uniform distribution on a compact interval, w.l.o.g. fixed to $[0,1]$, and is revealed during the campaign.  Candidates choose their platforms twice, both before and after the reveal of the median voter's position. We denote the ex-ante platforms $x_1$ and $x_2$, and the ex-post platforms $y_1$ and $y_2$. 

We interpret the choice of the ex-post platforms as the adjustment that candidates may make with respect to their previous positions after learning the electorate's preferences more precisely. The choice of the ex-ante platforms reflects instead the positions in which candidates invest over the long term, knowing that future revisions are allowed but costly.

Voter $t$'s utility from a victory of candidate $i$ depends on both the candidate's ex-ante and the ex-post platform according to the following functional form
\begin{equation} \label{Eq:utilityvoter}
u_t(x_i,y_i)=-(t-y_i)^2-a(y_i-x_i)^2.
\end{equation}
The first term represents voters' preference for a candidate whose final policy is closer to their ideal one. The second term represents how voters penalize candidates for changing their platform with respect to their previous position.  
In our model, voters trust that policy $y_i$ will be implemented if candidate $i$ wins, hence do not care about the distance between $t$ and $x_i$. Yet, they have an intrinsic preference for candidates who are consistent. The parameter $a>0$ measures the relative effect of this electoral penalty and thus how voters trade off policy considerations with their dislike for a candidate changing position. 

The election is decided by majority rule. Because voters' preferences are single-peaked and their utilities have the same functional form, the candidate that attracts the median voter wins the election.
The payoff of the candidates depends on both the outcome of the election and the possible organizational cost of a policy change. The benefit from winning the election is set to $1$ and that from losing to $0$.
In addition, a candidate that changed his policy, by selecting $y_i\neq x_i$, pays a fixed cost $\phi \in (0,\frac12)$.\footnote{We assume that the cost $\phi$ is smaller than $\frac12$, which is the expected gain from competing in the election since players are symmetric. If $\phi>0.5$, saving the organizational costs is more important than competing for the election.} We interpret $\phi$ as the cost induced by the need to communicate the change of policy to the public. Although greater changes may require more substantial communication strategies, we abstract from the dependency of the organizational cost on the magnitude of the change for technical simplicity. Our assumption of a fixed $\phi$ captures the presence of some fixed cost, which is likely present independently of the magnitude of the change.
The payoff of candidate $i$ is then
\[ g_i(\textbf{x},\textbf{y})=\E(\1_{i \text{ wins}}-\phi \1_{y_i\neq x_i}),\]
where $\1_{i\text{ wins}}$ is the indicator that candidate $i$ wins the election, and the expectation is taken with respect to the position of the median voter $m$ and the possibly mixed actions of the players.
Ties are broken by a toss of a fair coin, so in case of a tie, $\E(\1_{i\text{ wins}})=\frac12$.

To summarize, the timing of the game is as follows.
\begin{enumerate}
\item \textbf{First Stage:} The two candidates choose ex-ante platforms $x_1$ and $x_2$.
\item \textbf{Information Shock:} The position of the median voter $m$ is revealed.
\item \textbf{Second Stage:} The two candidates choose ex-post platforms $y_1$ and $y_2$.
\item \textbf{The Voting:} The candidate's preferred by the median voter is elected and candidates' payoffs are realized.
\end{enumerate}

\section{Equilibrium analysis}\label{Sect:MainResult}

We are interested in the subgame-perfect equilibrium of the game. Hence, we solve the game using backward induction and start by studying the strategic flip-flopping of candidates after their observation of the ex-ante platforms $x_1,x_2$ and the location of the median voter $m$ (Propositions~\ref{prop:second_stage} and~\ref{prop:r1=r2}). The analysis shows that two scenarios are possible in the second stage: either the advantage of one candidate is too large and he is sure to win, or the election is still open and the outcome depends on candidates' reactions. In the first case, it is optimal for both candidates not to change platforms and remain at their ex-ante positions. In the second case, both candidates mix between moving to an ex-post optimal position and not moving. 

Next, given the strategies in the second stage, we analyze the optimal choice of ex-ante platforms. We show that for sufficiently high values of the costs $(a,\phi)$, the unique subgame-perfect equilibrium requires candidates to invest in divergent ex-ante positions (Proposition~\ref{prop:first_stage}). Candidates differentiate in order to maximize the chances of having a sufficient advantage in the second stage, for which they save on the organizational cost of changing policy. In the remaining region of the parameter space (see Figure~\ref{fi:3regions}), a subgame-perfect equilibrium does not exist but an $\epsilon$-equilibrium exists in which candidates take centrist positions one $\epsilon$ away from each other (Proposition~\ref{prop:noEQinR1}). Finally, we summarize the implications in terms of flip-flopping behavior that emerge from our equilibrium analysis (Proposition~\ref{prop:implications} and~\ref{prop:comparative}).

\subsection{Second Stage: Choice of Ex-Post Platforms $y_1$, $y_2$}\label{Subsect:x_selection}
We take $x_1$, $x_2$ and $m$ as given and consider separately the cases where the ex-ante platforms are the same or different, as the analysis in these two cases differs significantly.

\subsubsection{Different Ex-Ante Platforms $x_1\neq x_2$}\label{311}
Let us assume first $x_1 \neq x_2$ and consider, without loss of generality, the case $x_1<x_2$. Given $m$, we refer to the candidate with ex-ante platform closer to $m$ as the \emph{favorite} candidate.

\begin{definition}\label{def:favorite}
Candidate $i$ is the \emph{favorite} and candidate $j$ is the \emph{challenger} if $|x_i-m|<|x_j-m|$.
\end{definition}

The intuition for the term \emph{favorite} is that, if candidates do not change their platforms, the favorite candidate $i$ wins the election as $u_m(x_i,x_i)>u_m(x_j,x_j)$. We ignore the case where $m=\frac{x_1+x_2}{2}$ which occurs with null probability. 

In the case where $m = x_1$, candidate~$1$ can win the election without changing platform and incurring the organizational cost $\phi$, as the median voter has the highest possible utility from his victory. Candidate~$2$, on the other hand, loses the election regardless of his ex-post platform, so it is optimal for both candidates to not move and save on the organizational cost. This argument remains true when $m \neq x_1$ but is close enough to $x_1$: candidate~$1$ still wins the election without changing his platform, and candidate~$2$ also does not change his platform as he loses the election anyway.
In this case, we say that the favorite candidate has \emph{secured} the election.

\begin{definition}\label{def:unquest_favorite}
The election is \emph{secured} for the favorite candidate $i$ if $u_m(x_i,x_i)>\max\limits_{y\in[0,1]} u_m(x_j,y)$.
Otherwise, the election is \emph{still open}.
\end{definition}

A favorite candidate who has secured the election has a strictly dominant strategy $y_i=x_i$, by which he wins without paying the organizational cost $\phi$.
The challenger also has a strictly dominant strategy $y_j=x_j$, as he cannot win the election and should save the organizational cost. Hence, if the the election is secured, the strategies $y_1=x_1$, $y_2=x_2$ constitute the unique equilibrium of the second stage subgame. 
The equilibrium payoffs are $1$ for the favorite candidate and $0$ for the challenger. In Lemma~\ref{cl:mmnn} in the appendix, we solve the inequality that appears in Definition~\ref{def:unquest_favorite} and show  that the election is secured for candidate~$1$ whenever $m\in (\underline{m},\overline{m})$, where
$$\underline{m}:= \tfrac{\a x_1 - x_2}{\a-1} \vee 0,\qquad\overline{m}:= \tfrac{\a x_1 + x_2}{\a+1}, $$
and where
\begin{equation}
\a=\sqrt{\tfrac{1+a}{a}}
\end{equation}
is a useful parameter.
Similarly, 
the election is secured for candidate~$2$ whenever $m\in (\underline{n},\overline{n})$ where 
$$\underline{n}:= \tfrac{\a x_2 + x_1}{\a+1}, \qquad \overline{n}:= \tfrac{\a x_2 - x_1}{\a-1} \wedge 1.$$
Figure~\ref{fig:rep_leadership} summarizes the regions where each candidate has secured the election, and those in which the election is still open.
 We observe that the election is secured if $m$ realizes in a neighborhood of each candidate's ex-ante platform $x_i$.

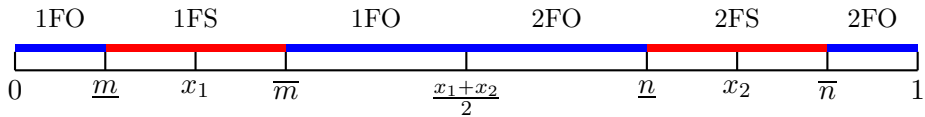
\begin{figure}[h]
\centering
\begin{tikzpicture}[x=1.2cm]
\draw[black,thick,>=latex]
  (0,0) -- (10,0);

\node[below,align=left,anchor=north,inner xsep=0pt]
  at (0,0) {$0$};	
\node[below,align=left,anchor=north,inner xsep=0pt]
  at (1,0) {$\underline{m}$};
\node[below,align=left,anchor=north,inner xsep=0pt]
  at (2,0) {$x_1$};
\node[below,align=left,anchor=north,inner xsep=0pt]
  at (3,0) {$\overline{m}$};
\node[below,align=left,anchor=north,inner xsep=0pt]
  at (5,0) {$\tfrac{x_1+x_2}{2}$};
\node[below,align=left,anchor=north,inner xsep=0pt]
  at (7,0) {$\underline{n}$};
\node[below,align=left,anchor=north,inner xsep=0pt]
  at (8,0) {$x_2$};
\node[below,align=left,anchor=north,inner xsep=0pt]
  at (9,0) {$\overline{n}$};
\node[below,align=left,anchor=north,inner xsep=0pt]
  at (10,0) {$1$};

\foreach \Xc in {0,1,2,3,5,7,8,9,10}
{
  \draw[black,thick]
    (\Xc,0) -- ++(0,7pt) node[above] {};
}

 \fill[red]
    (1,0.25)
      rectangle node[above, color=black] {\strut\small 1FS}
    (3,0.35);

 \fill[red]
    (7,0.25)
      rectangle node[above, color=black] {\strut\small 2FS}
    (9,0.35);

 \fill[blue]
    (0,0.25)
      rectangle node[above, color=black] {\strut\small 1FO}
    (1,0.35);

 \fill[blue]
    (3,0.25)
      rectangle node[above, color=black] {\strut\small 1FO}
    (5,0.35);

 \fill[blue]
    (5,0.25)
      rectangle node[above, color=black] {\strut\small 2FO}
    (7,0.35);

 \fill[blue]
    (9,0.25)
      rectangle node[above, color=black] {\strut\small 2FO}
    (10,0.35);

\end{tikzpicture}
\caption{An illustration of the status of each player being the favorite while the election is secured (FS) or still open (FO)  
before choosing the ex-post platform, with respect to the position of the median voter $m \in [0,1]$. 
} \label{fig:rep_leadership}
\end{figure}

If the election is still open, the favorite candidate wins the election if neither candidate changes his platform. However, he can lose the election if he keeps the platform while the challenger moves closer to the median voter. Hence, the challenger wants to move if the favorite does not change his platform. Analogously, the favorite also wants to move if the challenger moves, in order to secure his victory. In Lemma~\ref{cl:xstar} in the appendix, we show that in this case both candidates have an optimal platform $\hat{y_i}$ which they want to adopt if they decide to change their ex-ante platform. The optimal platform is given by

\begin{equation}\label{Eq:opt_plat}
\hat{y_i}=\frac{m+ax_i}{1+a}.
\end{equation}

This optimal platform is a weighted average between the ex-ante platform $x_i$ and the realization of the median voter location $m$. For each candidate, the platform $\hat{y_i}$ is optimal in the sense that all other platforms different from $x_i$ are either 
dominated or redundant, as proved also in Lemma~\ref{cl:xstar}. 
We can represent the ex-post game between the two candidates by the following $2\times 2$ one-shot game, assuming that each candidate chooses between moving to $\hat{y_i}$ and not moving. Without loss of generality, we consider candidate~$1$ as the favorite.

\begin{table}[h!]
\begin{center}
\begin{game}{2}{2}[Favorite candidate~$1$][Challenger candidate~$2$]
& $\hat{y_2}$ & $x_2$ \\
$\hat{y_1}$ &$(1-\phi,-\phi)$ &$(1-\phi,0)$\\
$x_1$ &$(0,1-\phi)$ &$(1,0)$
\end{game}
\caption{The normal form game that candidates face if the election is still open.}\label{tbl:2x2when_no_UQF}
\end{center}
\end{table}
If both candidates move to their optimal platforms, candidate~$1$ wins the election: his ex-ante platform is closer to the median voter, so his ex-post platform is both closer to the median voter and requires a smaller adjustment $\vert \hat{y_1}-x_1 \vert <\vert \hat{y_2}-x_2 \vert$. 
We observe that the favorite candidate wants to take the same action as the challenger, while the challenger wants to take the opposite action as the favorite. 
Hence, such a game has no equilibrium in pure strategies but has a unique equilibrium in mixed strategies. At this equilibrium, the favorite changes his platform with probability $(1-\phi)$ and the challenger changes his platform with probability $\phi$. The expected equilibrium payoffs are $1-\phi$ for the favorite candidate and $0$ for the challenger. 

This discussion is summarized in the following proposition.


\begin{proposition}[Equilibrium in the second stage subgame for $x_1 \neq x_2$]\label{prop:second_stage}
Suppose that different ex-ante platforms $x_1\neq x_2$ were chosen. The reaction of the candidates to the revelation of $m$ is:
\begin{itemize}
\item If the favorite candidate has secured the election, the unique subgame equilibrium is $(y_1,y_2)=(x_1,x_2)$. The equilibrium payoffs are $1$ for the favorite and $0$ for the challenger.
\item If the election is still open, the unique subgame equilibrium is in mixed strategies:
$$y_i= \begin{cases}  \hat{y_i} &\text{ with probability } 1-\phi, \\  x_i &\text{ with probability } \phi, \end{cases}$$ 
for the favorite candidate $i$ and
$$y_j=\begin{cases}\hat{y_j} &\text{ with probability } \phi,\\  x_j &\text{ with probability } 1-\phi, \end{cases}$$ for the challenger $j$. The equilibrium payoffs are $(1-\phi)$ for the favorite and $0$ for the challenger. 
\end{itemize}
\end{proposition}
\begin{proof}
See Appendix~\ref{Subse:second_stage}.\hfill\end{proof}

If the the election is secured, the uniqueness of the equilibrium is clear, as $(x_1,x_2)$ are strictly dominant strategies. 
Otherwise, the uniqueness concerns the equilibrium payoff but not the equilibrium strategies, as in some configurations other strategies might be redundant to the strategies $\hat{y_i}$ and constitute an equilibrium with the same payoffs. As shown in Lemma~\ref{cl:xstar}, disregarding these strategies is without loss of generality because it does not impact the payoffs nor the analysis.

The comparison of the subgame equilibrium payoffs when the election is secured or open illustrates the inefficiency of the second scenario.
When the election is still open, the favorite gets $(1-\phi)$ and the challenger $0$ (in expectation), which is the same payoffs as if the challenger were to concede but the favorite candidate still paid the organizational cost to change his policy.

\subsubsection{Identical Ex-Ante Platforms $x_1=x_2$}
In the case of identical ex-ante platforms, the election is always open and there is neither a favorite nor a challenger, which breaks down the analysis from Proposition~\ref{prop:second_stage}.
As before, both candidates have an optimal platform $\hat{y_i}$ where they want to move if they do. 
This optimal position is again equal to $\hat{y_i}=\tfrac{m+ax_i}{1+a}$ and is now identical for both candidates. Hence, if both candidates move or both do not move, each has a probability of $\tfrac{1}{2}$ to win the election. We can then restrict our attention to the ex-post game given by the $2 \times 2$ matrix in Table~\ref{tbl:2x2when_r1=r2}.

\begin{table}[h!]
\begin{center}
\begin{game}{2}{2}[Candidate~$1$][Candidate~$2$]
& $\hat{y_2}$ & $x_2$ \\
$\hat{y_1}$ &$(\frac12-\phi,\frac12-\phi)$ &$(1-\phi,0)$\\
$x_1$ &$(0,1-\phi)$ &$(\frac12,\frac12)$
\end{game}
\caption{The normal form game that candidates face if the ex-ante platforms are identical.}\label{tbl:2x2when_r1=r2}
\end{center}
\end{table}

Since $\phi<\tfrac{1}{2}$, we have $\tfrac{1}{2}-\phi>0$ and $1-\phi>\tfrac{1}{2}$. It follows that keeping the ex-ante platform is a strictly dominated strategy and in the unique equilibrium both candidates choose platforms $\hat{y_1}=\hat{y_2}$. The above discussion proves the following Proposition:

\begin{proposition}[Equilibrium in the second stage of the game for $x_1 = x_2$]\label{prop:r1=r2}
Suppose that ex-ante, the identical platforms $x_1=x_2$ were chosen. The unique subgame equilibrium in the second stage is $(\hat{y_1},\hat{y_2})$ (given by Eq.~\eqref{Eq:opt_plat}). The equilibrium payoffs are $(\frac12-\phi,\frac12-\phi)$.
\end{proposition}

Note that $\frac12-\phi$ is the worst possible equilibrium payoff in this game. Indeed, no matter the strategy of his opponent, candidate $i$ can select $x_i=\frac12$ and $y_i=\hat{y_i}$.
By doing so he wins the election with probability of at least $\frac12$, while paying the organizational cost $\phi$.

\subsection{First Stage: Choice of Ex-Ante Platforms $x_1$, $x_2$}\label{Subsect:r_selection}
By backward induction, candidates choose their ex-ante platforms by considering that the subgame equilibrium given by Proposition~\ref{prop:second_stage} or~\ref{prop:r1=r2} is played in the second stage. In the first stage, we restrict attention to pure strategies $x_i$ for both candidates. 

A first result is that there cannot exist a subgame-perfect equilibrium in which candidates select the same position in the first stage. If the ex-ante positions were identical, each candidate would have a profitable deviation by playing an infinitesimally different ex-ante position, since such a deviation induces a better equilibrium payoff after the second stage. 
We prove this claim within the proof of Proposition~\ref{prop:noEQinR1} below but anticipating this result allows us to focus on ex-ante platforms that properly define a favorite and a challenger, as in Section~\ref{311}. 

The next proposition shows that if $\phi$ and $a$ are sufficiently big, then the ex-ante platforms diverge from the center in the unique subgame-perfect equilibrium. The intuition behind this result is that,  conditional on being the favorite, candidates prefer the election to be secured rather than still open, when the position of $m$ is revealed. Divergence occurs because the probability of a secured election is proportional to the distance between the ex-ante platforms $|x_2-x_1|$.
More precisely, given the unique subgame equilibrium described in Proposition~\ref{prop:second_stage}, each candidate $i$'s expected payoff  is:
\begin{multline}\label{eq:payoff}
1 \times \mathbb{P}(i \text{ is the favorite and the election is secured}) +\\
+ (1-\phi) \times \mathbb{P}(i \text{ is the favorite and the election is still open}) + 0
\end{multline}
where the probability $\mathbb{P}$ represents the randomization of the median voter's location and the $0$ is the expected payoff if $i$ is the challenger. On the one hand, player prefer to be the favorite,
which creates an incentive to move towards the opposite candidate in order to increase the likelihood of being the closest candidate to $m$. 
On the other hand, being the favorite is not the only concern of the candidates, because they also want to maximize the probability of the election being secured conditional on being the favorite. The fact that the probability that the election is secured for the favorite is proportional to the distance between ex-ante platforms $|x_2-x_1|$ creates an incentive to diverge. Indeed, by differentiating their platforms, candidates create ``secure electorates'' -- regions of the policy space which guarantee a victory of a candidate if the median voter is revealed in such a region without the need of a costly adjustment of platform. This centrifugal force is to be traded off with the centripetal force and prevails in equilibrium if the condition on $\phi$ in Proposition~\ref{prop:first_stage} holds.

\begin{proposition}[Differentiation of ex-ante platforms]\label{prop:first_stage}
Suppose that $\phi>\frac{1}{1+4\sqrt{a(1+a)}}$ and without loss of generality that $x_1 \leq x_2$.
In the unique subgame-perfect equilibrium, candidates' ex-ante platforms are
\begin{equation}
(x_1^*,x_2^*)=\left(\frac{1}{\a+1},\frac{\alpha}{\a+1}\right)
\end{equation}
where $\alpha=\sqrt{\frac{1+a}{a}}$, and the ex-post behavior is according to Propositions~\ref{prop:second_stage} and~\ref{prop:r1=r2}.\\
The equilibrium expected payoffs are
\begin{equation}\label{Eq:Equilibrium_internal_payoff}
g_1^*=g_2^*= \frac{1}{2} - \frac{\phi}{2} \left(\frac{\a-1}{\a+1}\right)^2
\end{equation}
\end{proposition}

\begin{proof}
See Appendix~\ref{Subsect:first_stage}.
\hfill\end{proof}

In these equilibrium payoffs, the first term represents the expected gain of the election and the second term represents the expected organizational cost. Indeed, in equilibrium the election is still open with probability $\left(\frac{\alpha-1}{\alpha+1}\right)^2$ and in this case each player deviates with probability $1-\phi$ or $\phi$ whether he is the favorite or not, so he pays the cost $\phi$ with probability $\frac{1}{2}\phi + \frac{1}{2}(1-\phi) =\frac{1}{2}$.

Figure~\ref{fig:eq} shows the intervals on the policy space in which the election is secured for the favorite candidate or still open at the subgame-perfect equilibrium, which depend only on the electoral cost $a$ through the parameter $\a$.

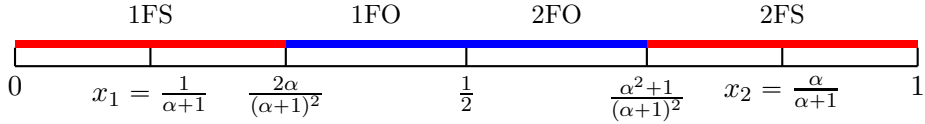
\begin{figure}[h!]
\centering
\begin{tikzpicture}[x=1.2cm]
\draw[black,thick,>=latex]
  (0,0) -- (10,0);

\node[below,align=left,anchor=north,inner xsep=0pt]
  at (0,0) {$0$};	
\node[below,align=left,anchor=north,inner xsep=0pt]
  at (1.5,0) {$x_1=\frac{1}{\a+1}$};
\node[below,align=left,anchor=north,inner xsep=0pt]
  at (3,0) {$\frac{2\alpha}{(\a+1)^2}$};
\node[below,align=left,anchor=north,inner xsep=0pt]
  at (5,0) {$\tfrac{1}{2}$};
\node[below,align=left,anchor=north,inner xsep=0pt]
  at (7,0) {$\frac{\a^2+1}{(\a+1)^2}$};
\node[below,align=left,anchor=north,inner xsep=0pt]
  at (8.5,0) {$x_2=\frac{\alpha}{\a+1}$};
\node[below,align=left,anchor=north,inner xsep=0pt]
  at (10,0) {$1$};

\foreach \Xc in {0,1.5,3,5,7,8.5,10}
{
  \draw[black,thick]
    (\Xc,0) -- ++(0,7pt) node[above] {};
}

 \fill[red]
    (0,0.25)
      rectangle node[above, color=black] {\strut\small 1FS}
    (3,0.35);

 \fill[red]
    (7,0.25)
      rectangle node[above, color=black] {\strut\small 2FS}
    (10,0.35);

 \fill[blue]
    (3,0.25)
      rectangle node[above, color=black] {\strut\small 1FO}
    (5,0.35);

 \fill[blue]
    (5,0.25)
      rectangle node[above, color=black] {\strut\small 2FO}
    (7,0.35);

\end{tikzpicture}
\caption{The status of each player being the favorite when the election is secured (FS) or still open (FO) at the subgame-perfect equilibrium described by Proposition~\ref{prop:first_stage}, with respect to the position of the median voter $m$.}
\label{fig:eq}
\end{figure}

The following proposition shows that if the condition on $a$ and $\phi$ in Proposition~\ref{prop:first_stage} is reversed, a subgame-perfect equilibrium does not exist. 
Instead, we prove the existence of an $\epsilon$-equilibrium, at which neither candidate can unilaterally improve his payoff by more than $\epsilon$, for every $\epsilon>0$.
\begin{proposition}\label{prop:noEQinR1}
If $0<\phi<\tfrac{1}{1+4\sqrt{a(1+a)}}$, there does not exist a subgame-perfect equilibrium. There exists instead an $\epsilon$-equilibrium, given by $(x_1,x_2)=(\tfrac{1}{2}-\epsilon,\tfrac{1}{2}+\epsilon)$ for every $\epsilon>0$ small enough.
\end{proposition}
\begin{proof}
See Appendix~\ref{Subsect:noEQinR1}.
\hfill\end{proof}

The intuition for the result is that if the organizational cost of changing platform is sufficiently small, candidates prefer to converge towards the center to increase the chances of being the favorite candidate, even if this decreases the probability of a secured election.
However, the centripetal incentive stops when candidates take the same position because of the discontinuity in the payoff when $x_1=x_2$. Indeed, if candidates have minimally differentiated platforms, each has an advantage in the second stage if the median voter is located on his side of the policy space. Formally, by converging fully to $x_1=x_2=\frac{1}{2}$ each candidate obtains an expected payoff equal to $\frac{1}{2}-\phi$. Instead, by minimally differentiating from the center and choosing $x_1=\frac{1}{2}-\epsilon$, $x_2=\frac{1}{2}+\epsilon$ they both obtain a higher expected payoff, which converges to $\frac{1}{2}-\frac{\phi}{2}$ as $\epsilon$ goes to $0$.

The condition on the parameters for the existence of a subgame-perfect equilibrium is drawn in Figure~\ref{fi:3regions} in the $(a,\phi)$ plane. The region dealt with in Proposition~\ref{prop:first_stage} is designated by $R_0$ and formally defined as $\{(a,\phi)\vert \Psi(a)<\phi<\tfrac{1}{2}\}$ where
\begin{equation}\label{Eq:Psi(a)}
\Psi(a)=\frac{1}{1+4\sqrt{a(1+a)}}.
\end{equation}
In the region named $R_1$, only an $\epsilon$-equilibrium exists.

\begin{figure}[h!]
\begin{center}
\begin{tikzpicture}

\begin{axis}[axis x line=middle, axis y line=middle, xmin=0,  xmax=3, ymin=0, ymax=0.5, xlabel={$a$}, ylabel={$\phi$},  x label style={at={(axis description cs:1.1,0)},anchor=north},    y label style={at={(axis description cs:0,1)},anchor=south}]
	\addplot[
		name path=A,
        domain = 0.058:3,
        samples = 100,
        smooth,
        thick,
    ] {(1/x)/(1/x+4*sqrt(1+(1/x)))};
	\addplot[name path=B, thick] coordinates{(0,0.5) (3,0.5)};
	\addplot[pattern=north west lines, pattern color=brown!50]fill between[of=A and B, soft clip={domain=0.059:10}];
	
	\node[] at (axis cs: 0.7,0.08) {$R_1$: No SPNE};  
	
	
	\node[] at (axis cs: 1.8,0.28) {$R_0$: SPNE $x_1\neq x_2$};
	
	\node[] at (axis cs: 2.3,0.055) {$\Psi(a)$};

\end{axis}
\end{tikzpicture}
\caption{The two regions $R_0$ and $R_1$ in the $a-\phi$ space. The colored region represents the area $R_0$ concerned with Proposition~\ref{prop:first_stage}.}
\label{fi:3regions}
\end{center}
\end{figure}
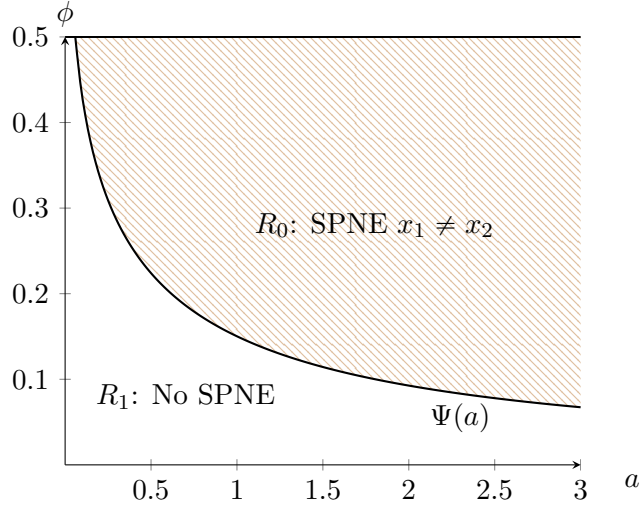

\subsection{Adjustments along the equilibrium path}
In light of our analysis, we can now put forward a few implications concerning the occurrence of candidates' policy changes. We focus on the parameter region in which the electoral and organizational costs are high enough, and thus a subgame-perfect equilibrium with divergent ex-ante positions exists.

\begin{proposition}\label{prop:implications}
Suppose that $\phi>\frac{1}{1+4\sqrt{a(1+a)}}$. At the subgame-perfect equilibrium,
\begin{enumerate}[(i)]
\item candidates flip-flop to voters' preferences only towards the center;\label{p5-1}
\item the favorite candidate is more likely to flip-flop than the challenger; \label{p5-2}
\item when the favorite candidate flip-flops, the magnitude of the adjustment is smaller than when the challenger flip-flops;\label{p5-3}
\item flip-flopping is always successful when done by a favorite candidate, while it is more likely to be unsuccessful than successful for a challenger.\label{p5-4}
\end{enumerate}
\end{proposition}
\begin{proof}
See Appendix~\ref{Subsect:prop_implications}.
\hfill
\end{proof}

Property (\ref{p5-1}) describes a dynamic of moderation along the electoral campaign, according to which candidates start out with more extreme positions and converge to the center if and once the median voter realizes in the center. The property also implies that candidates never cross over the position of the opponent. As such, even flip-flopping candidates always keep their relative ideological stands, despite the fact that they are office seekers and have no preferred policy. 

The intuition behind property (\ref{p5-2}) relies on the mixed strategies used by candidates at equilibrium: the favorite is indifferent between flip-flopping or not only when the probability to lose the election is relatively small, that is when his opponent is less likely to flip-flop. Instead, the disadvantaged candidate is indifferent only when the victory is relatively unlikely, that is when the favorite secures his advantage with high probability.

Property (\ref{p5-4}) means that a flip-flopping favorite always wins the election in equilibrium, while a flip-flopping challenger loses whenever the favorite also flip-flops, i.e. with probability $1-\phi>\frac{1}{2}$.
Taken together, these properties suggest that campaign flip-flopping consists most often in a minor adjustment by the favorite candidate in order to consolidate his victory, and only less likely it is a major and risky move by the challenger who tries to reverse the election outcome. 

The next proposition focuses on the comparative statics properties with respect to the electoral cost parameter $a$.

\begin{proposition}\label{prop:comparative}
When the electorate is more tolerant towards candidates changing positions, that is when $a$ decreases,
\begin{enumerate}[(i)]
\item the candidates' ex-ante policies are more polarized;\label{p6-1}
\item the election is more likely to be still open after the reveal of $m$, hence each candidate is more likely to flip-flop;\label{p6-2}
\item the candidates' equilibrium payoffs are lower.\label{p6-3}
\end{enumerate}
\end{proposition}

\begin{proof}
See Appendix~\ref{Subsect:prop_comparative}.
\hfill
\end{proof}

The comparative statics analysis highlights two main phenomena. First, candidates prefer an intransigent electorate that is less tolerant of flip-flopping. The underlying mechanism is described in claim (\ref{p6-2}): more tolerant electorates lead to increased competition, with more elections to be still open after the revelation of the median. As a result, winning elections without flip-flopping becomes less likely, leading to smaller equilibrium payoffs.\\
Additionally, while the electoral cost is a necessary ingredient in Proposition~\ref{prop:implications} to obtain differentiation, the degree of polarization decreases as the electoral cost increases. Indeed, in equilibrium, candidates push their secured intervals to the limits of the policy space, as $\underline{m}=0$ and  $\overline{n}=1$, and they position in the middle of these intervals. However, when $a$ is small, these intervals shrink due to increased competition (the challenger can more easily flip-flop to reverse the election), causing the center of these intervals to shift closer to the extremities, and the degree of candidates polarization is thus larger.


\section{Asymmetric candidates}\label{Sect:ASYM_candidates}
In this section, we relax the symmetry between candidates by supposing that voters penalize differently each candidate for changing platform, via different values of the parameter $a$.\footnote{A similar analysis can be made to consider different organizational costs $\phi$, without affecting the results significantly.}
Thus, the utility of voter $t$ from voting for candidate~$i$ with ex-ante and ex-post platforms $x_i$ and $y_i$ is 
\begin{equation}
u_t^i(x_i,y_i):=-(t-y_i)^2-a_i(y_i-x_i)^2
\end{equation}
as opposed to Eq.~\eqref{Eq:utilityvoter}.

An interesting phenomenon results from this heterogeneity between candidates: for a range of possible realizations of the median voter's location, the favorite candidate is not guaranteed to win the election even if he adjusts his position. Indeed, the optimal adjustments are still given by weighted averages between the ex-ante platforms and the realization of the median voter, namely $\hat{y_i}=\frac{m+a_ix_i}{1+a_i}$, but because the weights are different for each candidate, the challenger might attract the median voter after both candidates adjust. Suppose, for example, that $a_1<a_2$ and that platforms $x_1<x_2$ were chosen ex-ante. If $m$ realizes sufficiently close to $\frac{x_1+x_2}{2}$ then candidate~$1$ can adjust his position closer to $m$ than candidate~$2$, thanks to his smaller electoral cost, and win the election. This is true even if $m$ is larger than $\frac{x_1+x_2}{2}$, i.e. if candidate~$2$ is the favorite according to Definition~\ref{def:favorite}.
This justifies the following definition:

\begin{definition}
Candidate $i$ is a weak favorite and candidate $j$ is a strong challenger if $|x_i-m|<|x_j-m|$ but $u^i_{m}(x_i,\hat{y}_i)<u^j_{m}(x_j,\hat{y}_j)$.
\end{definition}

In other terms, a weak favorite candidate wins the election when no candidate adjusts, but looses the election when both candidates adjust their platforms. More precisely, the (second-stage) game played by a weak favorite and a strong challenger is as follows:

\begin{table}[h!]
\begin{center}
\begin{game}{2}{2}[Strong challenger candidate~$1$][Weak favorite candidate~$2$]
& $\hat{y_2}$ & $x_2$ \\
$\hat{y_1}$ &$(1-\phi,-\phi)$ &$(1-\phi,0)$\\
$x_1$ &$(0,1-\phi)$ &$(0,1)$
\end{game}
\end{center}
\end{table}

The game admits $(\hat{y_1},x_2)$ as its unique equilibrium, since both actions are strictly dominant. That is, in equilibrium, a strong challenger adjusts his position and wins the election with certainty, while a weak favorite does not move. The equilibrium payoffs are given by $(1-\phi,0)$ and are identical to the case in which candidate~$1$ is the favorite and the election is still open in the symmetric game.

For $a_1<a_2$, by solving $u^2_{m}(x_2,\hat{y}_2)<u^1_{m}(x_1,\hat{y}_1)$ with respect to $m$, we find that candidate~$2$ is a weak favorite when $m \in (\frac{x_1+x_2}{2},\tilde{m})$, in which $\tilde{m}=\frac{\a_1 x_2 + \a_2 x_1}{\a_1+\a_2}$ and $\a_i=\sqrt{\frac{1+a_i}{a_i}}$. This region is represented in light blue in Figure~\ref{fig:4}. 

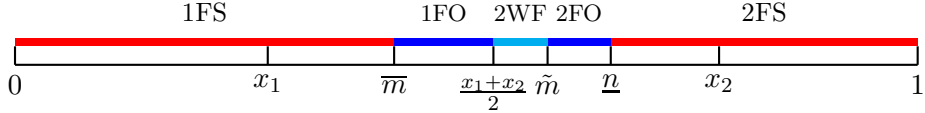
\begin{figure}[h!]
\centering
\begin{tikzpicture}[x=1.2cm]
\draw[black,thick,>=latex]
  (0,0) -- (10,0);

\node[below,align=left,anchor=north,inner xsep=0pt]
  at (0,0) {$0$};	
\node[below,align=left,anchor=north,inner xsep=0pt]
  at (2.8,0) {$x_1$};
\node[below,align=left,anchor=north,inner xsep=0pt]
  at (4.2,0) {$\overline{m}$};
\node[below,align=left,anchor=north,inner xsep=0pt]
  at (5.3,0) {$\frac{x_1+x_2}{2}$};
\node[below,align=left,anchor=north,inner xsep=0pt]
  at (5.9,0) {$\tilde{m}$};

  \node[below,align=left,anchor=north,inner xsep=0pt]
  at (6.6,0) {$\underline{n}$};
\node[below,align=left,anchor=north,inner xsep=0pt]
  at (7.8,0) {$x_2$};
\node[below,align=left,anchor=north,inner xsep=0pt]
  at (10,0) {$1$};

\foreach \Xc in {0,2.8,4.2,
5.3, 5.9,6.6,7.8,10}
{
  \draw[black,thick]
    (\Xc,0) -- ++(0,7pt) node[above] {};
}

 \fill[red]
    (0,0.25)
      rectangle node[above, color=black] {\strut\small 1FS}
    (4.2,0.35);

 \fill[red]
    (6.6,0.25)
      rectangle node[above, color=black] {\strut\small 2FS}
    (10,0.35);

 \fill[blue]
    (4.2,0.25)
      rectangle node[above, color=black] {\strut\footnotesize 1FO}
    (5.3,0.35);

 \fill[blue]
    (5.9,0.25)
      rectangle node[above, color=black] {\strut\footnotesize 2FO}
    (6.6,0.35);

 \fill[cyan]
    (5.3,0.25)
      rectangle node[above, color=black] {\strut\footnotesize 2WF}
    (5.9,0.35);

\end{tikzpicture}
\caption{The status of each player being the favorite when the election is secured (FS) or still open (FO) at the subgame-perfect equilibrium described by Proposition~\ref{prop:asym_a}, with respect to the position of the median voter $m$, for $a_1<a_2$.
In the light blue region, candidate~$2$ is the weak favorite (WF)}
\label{fig:4}
\end{figure}

The analysis of the second stage in all other regions (when the favorite is not weak) is the same in the asymmetric game as in the symmetric game of the previous sections.
The next proposition describes the subgame-perfect equilibrium considering both stages of the game.

\begin{proposition}[Equilibrium with asymmetric electoral costs]\label{prop:asym_a}
Assume that $a_1<a_2$ and that $\phi > \max(\Psi(a_1),\Psi(a_2))$. 
There exists a unique subgame-perfect equilibrium, in which the ex-ante locations are given by 
$$(x_1^*,x_2^*)=\left(\frac{\a_1-1}{\a_1\a_2-1},\frac{\a_2(\a_1-1)}{\a_1\a_2-1}\right)$$
where $\a_i=\sqrt{\frac{1+a_i}{a_i}}$.

The ex-post behavior depends on $m$:
\begin{itemize}
\item For $m\in \left[\tfrac{x^*_1+x^*_2}{2},\tilde{m}\right]$ where $\tilde{m}=\tfrac{\a_1x_2^*+\a_2x_1^*}{\a_1+\a_2}=\tfrac{\a_2(\a_1^2-1)}{(\a_1\a_2-1)(\a_1+\a_2)}$, candidate~$1$ adjusts to $\hat{y}_1$ and wins the election, whereas candidate~$2$ remains at $x_2^*$ and loses the election.
\item Otherwise, it follows Proposition~\ref{prop:second_stage}.
\end{itemize}

If $\phi < \max(\Psi(a_1),\Psi(a_2))$, there is no subgame-perfect equilibrium.
\end{proposition}

\begin{proof} See Appendix~\ref{Subsect:asym_a}\end{proof}

While the ex-ante positions were symmetric around $\tfrac{1}{2}$ in the case of identical electoral costs, we now have that 
\begin{equation}\label{diff}
\frac{x_1^*+x_2^*}{2}=\frac{1}{2}+\frac{\a_1-\a_2}{2(\a_1\a_2-1)}
\end{equation}
hence, the candidate with a lower $a$ (a higher $\a$) takes a more centrist ex-ante position and has a higher expected payoff. Expected payoffs in equilibrium are modified with respect to the expressions in Eq.~\eqref{Eq:Equilibrium_internal_payoff} by the possibility for the favorite to be weak, as follows:
\begin{equation}\label{payoffs}
\begin{aligned}
&g_1= 1 \times \left[\tfrac{2\a_2(\a_1-1)}{(\a_1 \a_2 -1)(\a_2+1)}\right] + (1-\phi) \times \left[\tfrac{(\a_1-1)(\a_1\a_2+\a2)}{(\a_1\a_2-1)(\a_1+\a_2)}-\tfrac{2\a_2(\a_1-1)}{(\a_1 \a_2 -1)(\a_2+1)}\right]\\
&g_2=1\times \left[1-\tfrac{(\a_1\a_2+1)(\a_1-1)}{(\a_1\a_2-1)(\a_1+1)} \right]+(1-\phi)\times \left[\tfrac{(\a_1\a_2+1)(\a_1-1)}{(\a_1\a_2-1)(\a_1+1)}-\tfrac{\a_2(\a_1^2-1)}{(\a_1\a_2-1)(\a_1+\a_2)} \right]
\end{aligned}
\end{equation}

With respect to Proposition~\ref{prop:implications}, if candidates are asymmetric, it is still true that flip-flopping happens only towards the center. It is no longer true, in general, that a favorite candidate is more likely to flip-flop than a challenger or that the magnitude of his adjustment is smaller. The validity of these statements depends now on whether the favorite is weak or not and on the exact difference in the electoral costs $a_i$. Flip-flopping is now always successful both for a favorite candidate (recall that a weak favorite never flip-flops in equilibrium) and for a strong challenger, while it remains more likely unsuccessful than successful for a challenger who is not strong.

The introduction of heterogeneous electoral costs also allows us to refine the comparative statics properties in Proposition~\ref{prop:comparative}. We can deduce from Eq.~\eqref{diff} and Eq.~\eqref{payoffs} that, when voters are more tolerant towards a change of position by candidate $i$ (i.e. if $a_i$ decreases for a given $a_j$), the ex-ante position of $i$ is more centrist and his expected payoff is higher, while the ex-ante position of $j$ is more extreme and his expected payoff is lower. Hence, while in the symmetric case, candidates are better off when the electorate penalizes flip-flopping more, in the asymmetric case candidates are better off only if a higher electoral cost concerns the opponent. In principle, a higher electoral cost could be beneficial for a candidate by reducing the likelihood of adjusting the position and paying the organizational cost. We find instead that, in equilibrium, each candidate would prefer voters to be more tolerant with himself. Indeed, the region in which a candidate $i$ has a secured election is independent of $a_i$ and only depends on the electoral cost $a_j$ for the opponent. Hence, any increase in $a_i$ simply decreases the likelihood of $i$ being in a better position (non-weak favorite or strong challenger) in a still-open election, and it is therefore detrimental.

\appendix
\numberwithin{equation}{section}
\section{Proofs}\label{app:proofs}

\subsection{Useful Lemmas}\label{app:lemmas}
In this section we present results which we use in the paper. Most of these results are simple computation and are omitted from the main text for the ease of reading.

\begin{lemma}\label{cl:mmnn}
Suppose $x_1 \neq x_2$ and without loss of generality $x_1 < x_2$. Candidate~$1$ has secured the election if and only if $ m  \in (\underline{m},\overline{m})$ with 
$$\underline{m}:= \tfrac{\a x_1 - x_2}{\a-1} \vee 0$$
$$\overline{m}:= \tfrac{\a x_1 + x_2}{\a+1} $$
candidate~$2$ has secured the election if and only if $m  \in (\underline{n},\overline{n})$ with 
$$\underline{n}:= \tfrac{\a x_2 + x_1}{\a+1} $$
$$\overline{n}:= \tfrac{\a x_2 - x_1}{\a-1} \wedge 1$$
\end{lemma}

\begin{proof} 
candidate~$1$ has secured the election when $u_m^1(x_1,x_1)>\max\limits_y u_m^2(x_2,y)$. 
After simplification, we find that it is the case when:
$$(m-x_2)^2 > \frac{1+a}{a} (m-x_1)^2$$
A case by case analysis (whether $m<x_1$, $m \in [x_1,x_2]$ or $m>x_2$ gives that candidate~$1$ has secured the election if $m \in (\tfrac{\a x_1 - x_2}{\a-1},\tfrac{\a x_1 + x_2}{\a+1})$.
Because $m$ has a support in $[0,1]$ that is $m \in (\underline{m},\overline{m})$.
The computations are symmetric for candidate~$2$.
We easily verify that $\underline{m} < x_1  < \overline{m} < \tfrac{x_1+x_2}{2} < \underline{n}< x_2 < \overline{n}$.
\hfill \end{proof}

\begin{lemma}\label{cl:xstar}  
Consider the ex-post platform that maximizes the utility of the median voter from candidate $i$:
\[\hat{y_i}=\argmax_{y_i} u_m (x_i,y_i)=\tfrac{m+ax_i}{1+a}\]


We have:
\begin{enumerate}
\item If candidate $i$ is the favorite, any platform $y_i \notin \{x_i,\hat{y_i}\}$ is either dominated by or redundant with a strategy in $\{x_i,\hat{y_i}\}$. 
\item If candidate $i$ is the challenger, any platform $y_i \notin \{x_i,\hat{y_i}\}$ is weakly-dominated by the platform $\hat{y_i}$ .
\item Weakly dominated actions are not played in equilibrium.
\end{enumerate}
Hence, at equilibrium, candidates select $\hat{y_i}$ when they adjust their platforms.
\end{lemma}
\begin{proof}
Suppose w.l.o.g. that candidate~$1$ is the favorite.
For any pair of ex-ante platforms $(x_1,x_2)$, we define $W_1 := \{y_1 | u_m^1(x_1,y_1) > \displaystyle\max\limits_{y_2} u_m^2(x_2,y_2) \}$, the set of policies of candidate~$1$ that guarantee him to win the election.

1.  By definition, $\hat{y_1} \in W_1$. In addition, 
If $x_1 \in W_1$, any action $y_1 \neq x_1$ is strictly dominated by $x_1$ because $x_1$ yields a payoff of $1$ and any other actions yields at most $1-\phi$, hence it is strictly dominated by $x_1$.

If $x_1 \notin W_1$, any action $y_1 \in W_1, y_1 \neq \hat{y_1}$ is redundant with $\hat{y_1}$ and any action $y_1 \notin W_1, y_1 \neq x_1$ is weakly dominated by $\hat{y_1}$. 
Indeed, consider $y\in W_1\setminus \{\hat{y_1}\}$. Both $y$ and $\hat{y_1}$ ensure a victory with payoff $1-\phi$ regardless of the platform of candidate~$2$. Hence, $y$ is redundant with $\hat{y_1}$. 

Consider instead $y\notin W_1, y\neq x_1$.
Consider strategy $\sigma$ that chooses $y$ with positive probability and strategy $\sigma'$ which is identical to $\sigma$, except that it plays $\hat{y_1}$ instead of $y$. Then the payoff of candidate~$1$ with $\sigma'$ is weakly higher than with $\sigma$ regardless of the strategy of candidate~$2$. Indeed, the expected payoff when using $y$ includes a term of the form $p-\phi$, where $p$ is the probability to win when using $y$ against the strategy of candidate~$2$, while the corresponding term when using $\hat{y_1}$ is $1-\phi$, as $\hat{y_1}$ wins with probability 1. The other terms in the expected payoff remain the same.

2. Any strategy that chooses $y\notin \{\hat{y_2},x_2\}$ with some positive probability is weakly dominated by a strategy that transfers this probability to $\hat{y_2}$. In both cases, $\phi$ is paid, but $u^2_m(x_2,\hat{y_2})>u^2_m(x_2,y)$ so $\hat{y_2}$ wins in all the events in which $y$ wins (and possibly in other events), resulting in a weakly higher payoff.

3. Weakly dominated action cannot be played in equilibrium. 
Suppose $y\notin W_1, y\neq x_1$ and suppose that in equilibrium, candidate~$1$ plays $y$ with positive probability.
Then his expected payoff when choosing $y$ (which, by indifference, is his equilibrium payoff) is $p-\phi$ where $p$ is the probability of winning the election, $p\leq 1$.
If $p<1$, then $\hat{y_1}$ is a profitable deviation of candidate~$1$ as he obtains $1-\phi$. If $p=1$, candidate~$2$ always looses the election and his expected payoff is non-positive. Then in equilibrium candidate~$2$ must choose $y_2=x_2$ which grants a payoff of $0$, and candidate~$1$ again has a profitable deviation to $y_1=x_1$ which yields a payoff of $1$. This contradicts the assumption that $y$ is played with positive probability in equilibrium.
\hfill\end{proof}

\subsection{Proof of Proposition~\ref{prop:second_stage}}\label{Subse:second_stage}
If the favorite candidate has secured the election, he has a strictly dominant strategy $y_i=x_i$ whose best reply for the challenger is $y_j=x_j$. Hence this is the unique subgame equilibrium. 

If the election is still open, by Lemma~\ref{cl:xstar}, we can limit the analysis to the $2 \times 2$ game in Table~\ref{tbl:2x2when_no_UQF}. The game has no pure strategy equilibria. The unique mixed strategy equilibrium is such that both candidates are indifferent between the two pure strategies, given the mixed strategy of the opponent. Hence the probability $p$ of playing $\hat{y_i}$ for the favorite candidate solves
\[
-p\phi+(1-p)(1-\phi)=0 \quad \Rightarrow \quad p=1-\phi,
\]
and the probability $q$ of playing $\hat{y_i}$ for the challenger candidate solves
\[
(1-\phi)=(1-q) \quad \Rightarrow \quad q=\phi.
\]

\subsection{Proof of Proposition~\ref{prop:first_stage}}\label{Subsect:first_stage}

We perform a backward analysis by supposing that after choosing actions $x_1,x_2$, players play the second-stage equilibrium provided in Proposition~\ref{prop:second_stage} and Proposition~\ref{prop:r1=r2}.

We first show that in equilibrium, the best response to a certain ex-ante platform is never a more extreme platform to the same side.
Formally and without loss of generality,
if $x_2<\frac12$ then $BR_1(x_2) \geq x_2$, and if
 $x_2>\frac12$ then $BR_1(x_2) \leq x_2$,
where $BR_1$ is the best response ex-ante platform of Candidate~$1$ to the ex-ante platform $x_2$ of Candidate~$2$.

To prove that, we show that if $x_1<x_2<\frac12$ then $g_1(x_1,x_2)<g_1(1-x_1,x_2)$.
By the uniform distribution of $m$, $g_1(x_1,x_2)=g_1(1-x_1,1-x_2)$. 
Hence, proving $g_1(x_1,x_2)<g_1(1-x_1,x_2)$ is equivalent to proving that $g_1(x_1,x_2)<g_1(x_1,1-x_2)$.
When $x_1<x_2<\tfrac{1}{2}$ and Candidate~$2$ moves from $x_2$ to $1-x_2$:
\begin{itemize}
\item candidate~$1$ is the favorite more often as $\tfrac{x_1+x_2}{2}$ increases with $x_2$.
\item candidate~$1$ secures the election more often as both $\overline{m}$ increases and $\underline{m}$ decreases with $x_2$.
\end{itemize}
We conclude that the expected payoff $g_1$ of candidate~$1$ is larger, given $g_1= 1 \times \mathbb{P}(1 \text{ FS })+ (1-\phi) \times \mathbb{P}(1 \text{ FO })$.
We can therefore suppose that $x_2 \geq \tfrac{1}{2}$ and conclude that the best response of Candidate~$1$ belongs to $[0,x_2]$.\footnote{In the border case where $x_2=\frac12$, both sides of $x_2$ are symmetric and for every best response $x_1$ in $[x_2,1]$, the ex-ante platform $1-x_1 \in [0,x_2]$ is also a best response.} The case where $x_2 \leq \tfrac{1}{2}$ and $x_2 \leq x_1$ is symmetric.

For $x_1<x_2$, in each region of the graph in Figure~\ref{fig:rep_leadership}, the payoff is determined by Proposition~\ref{prop:second_stage}:\footnote{We can safely ignore null probability events regarding the boundaries, such as $m=\overline{m}$.}
\begin{itemize}
\item If $m \in (\underline{m},\overline{m})$, candidate~$1$ has secured the election and his second-stage equilibrium payoff is $1$.
\item If $m \in [0,\underline{m}) \cup (\overline{m},\frac{x_1+x_2}{2})$, candidate ~$1$ has secured the election and his second-stage equilibrium payoff is $1-\phi$.
\item If $m \in (\frac{x_1+x_2}{2},1]$, candidate~$1$ is the challenger and his second-stage equilibrium payoff is $0$.
\end{itemize}

Because $m$ is uniformly drawn on the unit interval, the ex-ante payoff of candidate~$1$ is
\begin{equation*}
g_1(x_1,x_2)= (\overline{m}-\underline{m})+(1-\phi)\left(\tfrac{x_1+x_2}{2}-\overline{m}+\underline{m}-0\right) 
\end{equation*}
Substituting the values of $\overline{m}$ and $\underline{m}$ from Lemma~\ref{cl:mmnn}, we have that for $x_1<x_2$
\[
g_1(r,x)=\begin{cases}
x_1(\tfrac{1-\phi}{2}+\tfrac{\phi\a}{\a+1})+x_2(\tfrac{1-\phi}{2}+\tfrac{\phi}{\a+1}) \mbox{ if } 0 \leq x_1 \leq \tfrac{x_2}{\a} \\
x_1(\tfrac{1-\phi}{2}-\tfrac{2\phi \a}{\a^2-1})+x_2(\tfrac{1-\phi}{2}+\tfrac{2\phi \a}{\a^2-1})   \mbox{ if } \tfrac{x_2}{\a}<x_1 < x_2
\end{cases}
\]
The payoff function $g_1$ is therefore increasing with $x_1$ in the region $[0,\frac{x_2}{\a}]$.
In the region $[\tfrac{x_2}{\a},x_2)$, the function $g_1$ decreases with $x_1$ when  $\tfrac{1-\phi}{2}-\tfrac{2\phi \a}{\a^2-1}<0$, which is equivalent to  
\[
\phi > \tfrac{\a^2-1}{\a^2+4\a-1}=\tfrac{1}{1+4\sqrt{a(1+a)}}=\Psi(a)
\]
It follows that the best response to $x_2$ for candidate~$1$ in the region $[0,x_2)$ is $x_1^*(x_2)=\tfrac{x_2}{\a}$.
If instead $x_1=x_2$, then according to Proposition~\ref{prop:r1=r2} the payoff $g_1$ is $\tfrac{1}{2}-\phi$, regardless of $x_2$. 
Note that by choosing $x_1=x_2-\epsilon$, the payoff $g_1$ would be $x_2(1-\phi)$, which is always greater than $\tfrac{1}{2}-\phi$ for any $x_2\geq \tfrac{1}{2}$ and $\phi>0$. Hence $x_1=x_2$ cannot be a best reply for candidate~$1$.

The same arguments apply to candidate~$2$ for any $x_1\leq \frac{1}{2}$.
The function $g_2$ is increasing with $x_2$ in the $[x_1,1-\tfrac{1-x_1}{\a}]$ region and decreasing with $x_2$ in the $[1-\frac{1-x_1}{\a},1]$ region, when $\phi >\Psi(a).$
Hence, the best response to $x_1$ for candidate~$2$ in the region $(x_1,1]$ is $x_2^*(x_1)=1-\tfrac{1-x_1}{\a}$. Choosing $x_2=x_1$ cannot be optimal as it is dominated by  $x_2=x_1+\epsilon$.

To conclude, when $\phi>\Psi (a)$ the function $g_1$ admits a global maximum in $x_1^*(x_2)=\tfrac{x_2}{\a}$ and the function $g_2$ admits a global maximum in $x_2^*(x_1)=1-\tfrac{1-x_1}{\a}$.
In equilibrium, both candidates best respond to each other, and these two equations provide the unique profile: $x_1^*=\tfrac{1}{\a+1}$ and $x_2^*=\tfrac{\a}{\a+1}$.
The payoffs associated with these locations are:
\[
g_1^*=g_2^* =\frac12-\frac{\phi}{2}\left( \frac{\a-1}{\a+1}\right)^2
\]


\subsection{Proof of Proposition~\ref{prop:noEQinR1}}\label{Subsect:noEQinR1}

According to the proof of Proposition~\ref{prop:first_stage}, since $\phi<\Psi(a)$, the payoff $g_1$ is strictly increasing with $x_1$ in the $[0,x_2)$ region. 
The unique possible equilibrium is then $x_1=x_2$, where the payoff is discontinuous.
At the limit $x_1 \rightarrow x_2$, both $\underline{m}$ and $\overline{m}$ converge to $x_2$, so: 
\[\lim\limits_{x_1\to x_2} g_1(x_1,x_2)= \lim\limits_{x_1\to x_2} \tfrac{x_1+x_2}{2} (1-\phi) + \phi(\overline{m} - \underline{m})=x_2(1-\phi)\]
On the other hand, the second stage equilibrium when $x_1=x_2$ yields a payoff of $g_1=\tfrac{1}{2} - \phi$.
The condition for equilibrium can then be written $x_2(1-\phi)\leq \tfrac{1}{2}-\phi$ or equivalently $x_2\leq \tfrac{\tfrac{1}{2}-\phi}{1-\phi}$.
Because $0<\phi<\tfrac{1}{2}$ the right hand side of the previous inequality is smaller than $\tfrac{1}{2}$, so $x_2<\tfrac{1}{2}$. 
By repeating the same calculation for candidate~$2$, we obtain that $x_1=x_2$ is an equilibrium only when $x_1> \tfrac{1}{2}$, which cannot hold with $x_1=x_2$ and $x_2<\tfrac{1}{2}$.

We now prove that if $\phi<\Psi(a)$, $(x_1,x_2)=(\frac{1}{2}-\epsilon,\frac{1}{2}+\epsilon)$ is an $\epsilon$-equilibrium. Based on the expression computed in subsection~\ref{Subsect:first_stage}, candidate~$1$'s payoff is  $g_1(\frac{1}{2}-\epsilon,\frac{1}{2}+\epsilon)=(\frac{1}{2}-\epsilon)(\tfrac{1-\phi}{2}-\tfrac{2\phi \a}{\a^2-1})+(\frac{1}{2}+\epsilon)(\tfrac{1-\phi}{2}+\tfrac{2\phi \a}{\a^2-1})=\frac{1-\phi}{4}+\epsilon \frac{4 \alpha \phi}{\alpha^2-1}$. On the other hand, because $\phi<\Psi(a)$,~$g_1(x_1,\frac{1}{2}+\epsilon)$ is increasing with $x_1$ in the region $[0,\frac{1}{2}+\epsilon)$, therefore we can bound from above the possible payoff of candidate~$1$ by $\displaystyle\lim_{x_1 \uparrow \frac{1}{2}+\epsilon}g(x_1,\frac{1}{2}+\epsilon)=\frac{1-\Phi}{4}+\epsilon \frac{1-\phi}{2}$. The loss of candidate~$1$ when playing $\frac12-\epsilon$ is
\begin{align*}
\lim\limits_{x_1 \rightarrow \frac12-\epsilon}g_1(x_1,\frac12+\epsilon) - g_1(\frac12-\epsilon,\frac12+\epsilon)
&=  \frac{1-\phi}{4}+\epsilon \frac{1-\phi}{2} -  \frac{1-\phi}{4}+\epsilon \frac{4 \alpha \phi}{\alpha^2-1}\\
&= \epsilon \left( \frac{1-\phi}{2}-\frac{4\alpha \phi}{\alpha^2-1}\right)\rightarrow 0 \text{ as } \epsilon \rightarrow 0.
\end{align*}

\subsection{Proof of Proposition~\ref{prop:implications}}\label{Subsect:prop_implications}
(\ref{p5-1}) Based on Proposition~\ref{prop:second_stage} and~\ref{prop:first_stage}, the election is open if and only if $m \in \left[\frac{2\a}{(\a+1)^2},\frac{\a^2+1}{(\a+1)^2}\right]$. On the other hand, $x_1^*$ is on the left of this interval and $x_2^*$ on the right. Because candidate $i$ eventually flip-flops to $\hat{y_i}=\frac{m+ax_i^*}{1+a}$, he flip-flops towards the center.

\noindent (\ref{p5-2}) At the mixed equilibrium played in the second stage, if the election is still open, the favorite candidate flip-flops with probability $1-\phi>\frac12$ and that the challenger flip-flops with probability $\phi<\frac12$.

\noindent (\ref{p5-3}) Based on Eq.~\eqref{Eq:opt_plat}, the optimal platforms $\hat{y_1}$ and $\hat{y_2}$ are the same weighted average of $m$ and respectively $x_1$ and $x_2$. If player $1$ is the favorite, $x_1$ is closer to $m$ than $x_2$. Therefore, the distance between $\hat{y_1}$ and $x_1$ is smaller than the distance between $\hat{y_2}$ and $x_2$.

\noindent (\ref{p5-4}) By point (\ref{p5-2}) and the fact that a favorite candidate $i$ always wins if he moves to $\hat{y_i}$.

\subsection{Proof of Proposition~\ref{prop:comparative}}\label{Subsect:prop_comparative}
\noindent (\ref{p6-1}) Based on Proposition~\ref{prop:first_stage}, equilibrium location are given by $\left(\frac{1}{1+\a},\frac{\a}{1+\a}\right)$. The distance between ex-ante platforms is therefore equal to $\frac{\a-1}{\a+1}$ which is increasing with $\a$, thus decreasing with $a$.

\noindent (\ref{p6-2}) Propositions~\ref{prop:second_stage} and~\ref{prop:first_stage} together prove that at equilibrium, the election is still open with probability $\left(\frac{\a-1}{\a+1}\right)^2$ which is also increasing with $\a$, thus decreasing with $a$. Conditionally on the election to be still open, the likelihood of flip-flopping does not depend on $a$ (it is $1-\phi$ and $\phi$ for the favorite and the challenger respectively).

\noindent (\ref{p6-3}) Proposition~\ref{prop:first_stage} gives that equilibrium payoffs are decreasing with the probability of the election to be still open $\left(\frac{\a-1}{\a+1}\right)^2$, so this claim is a corollary of claim (\ref{p6-2}).

\subsection{Proof of Proposition~\ref{prop:asym_a}}\label{Subsect:asym_a}

We follow the proof of the symmetric competition (Appendix~\ref{Subse:second_stage} and~\ref{Subsect:first_stage}) and discuss the differences. 
First, the optimal ex-post positions (given by Eq.~\eqref{Eq:opt_plat} in the symmetric case) differs with $a_i$:  
we find $\hat{y_i}=\tfrac{m+a_ix_i}{1+a_i}$. 
A candidate with a larger $a_i$ is more penalized by the voters and changes less his ex-ante platform than his opponent. 
Second, the interval of $m$ for which each candidate has secured the election generalizes to $(\underline{m},\overline{m})=(\tfrac{\a_2 x_1 - x_2}{\a_2-1} \vee 0,\tfrac{\a_2 x_1 + x_2}{\a_2+1})$ for candidate~$1$ and $(\underline{n},\overline{n})=(\frac{\a_1 x_2+x_1}{\a_1+1},\frac{\a_1 x_2-x_1}{\a_1-1}\wedge 1)$ for candidate~$2$, where $\a_i=\sqrt{\frac{1+a_i}{a_i}}$. 
Notice that the region where each candidate has secured the election does not depend on his own parameter $\a_i$ but on his opponent's.
The intuition is that the region where candidate $i$ has secured the election is the region where he wins the election without moving, so the penalty for his movement does not play a role. 
Moreover, the larger $a_j$, the smaller movement of candidate $j$ towards the median voter and therefore the region where candidate $i$ has secured the election is on a broader region of possible $m$.

Next, we solve the inequality $u_m^1(x_1,\hat{y}_1)>u_m^2(x_2,\hat{y}_2)$, which holds for $m<\tfrac{\a_1x_2+\a_2x_1}{\a_1+\a_2}:=\tilde{m}$.
Note that $\tilde{m}>\tfrac{x_1+x_2}{2}$ for $a_1<a_2$ with equality when $a_1=a_2$.
Hence, in the region $\left(\tfrac{x_1+x_2}{2},\tilde{m}\right)$ candidate~$2$ is favorite, the election is still open and candidate~$2$ cannot defend his advantage: when candidate~$1$ moves to $\hat{y}_1$, he wins the election regardless of the action of candidate~$2$.
Since candidate~$1$ loses the election if he does not move, moving is a dominating strategy and in this region the optimal strategy for candidate~$1$ is to move, for candidate~$2$ not to move and the payoff is $(1-\phi,0)$.

Note that the organizational costs are unchanged, so the second stage remains strictly identical to Proposition~\ref{prop:second_stage} for $m\notin \left( \tfrac{x_1+x_2}{2}, \tilde{m}\right)$. 
Moreover, the payoff in $\left( \tfrac{x_1+x_2}{2}, \tilde{m}\right)$ is exactly the same as the expected payoff in the region where candidate~$1$ is favorite and the election is still open, so, although the second stage strategy is different, in terms of continuation payoff we can analyze the first stage as if candidate~$1$ is favorite and the election is still open in the region $[\overline{m},\tilde{m}]$ instead of $[\overline{m},\tfrac{x_1+x_2}{2}]$.

Finally, repeating the argument in Appendix~\ref{Subsect:first_stage}, we find that it is necessary and sufficient to have $\phi>\Psi(a_1)$ to guarantee the existence of a best response of candidate~$1$, which is given by $x_1(x_2)=\tfrac{x_2}{\a_2}$, and analogously, we have $x_2(x_1)=1-\tfrac{1-x_1}{\a_1}$ when $\phi>\Psi(a_2)$.
Together, we conclude that if $\phi>\max(\Psi(a_1),\Psi(a_2))$, then $x_1^*=\frac{\a_1-1}{\a_1\a_2-1}$ and $x_2^*= \frac{\a_2(\a_1-1)}{\a_1\a_2-1}$.
If any of the thresholds is higher than the organizational cost, there exists no equilibrium.

\bibliographystyle{plainnat}
\bibliography{bib_costly_adjustments}

\end{document}